\documentclass[a4paper,UKenglish,cleveref,autoref,thm-restate]{lipics-v2021}

\usepackage{amsfonts}
\usepackage{latexsym,amssymb,amsmath}

\usepackage[bibliography=common,appendix=inline]{apxproof}

\newtheoremrep{mythm}[theorem]{Theorem}
\newtheoremrep{mylmm}[theorem]{Lemma}

\usepackage[ruled]{algorithm}
\usepackage[noend]{algpseudocode}
\algdef{SE}[DOWHILE]{Do}{doWhile}{\algorithmicdo}[1]{\algorithmicwhile\ #1}

\usepackage{todonotes}

\title{Line Search for an Oblivious Moving Target}

\author{Jared Coleman}{University of Southern California, USA \and \url{https://jaredraycoleman.com} }{jaredcol@usc.edu}{https://orcid.org/0000-0003-1227-2962}{}

\author{Evangelos Kranakis}{Carleton University, Ottawa, Ontario, Canada \and \url{https://people.scs.carleton.ca/~kranakis/} }{kranakis@scs.carleton.ca}{https://orcid.org/0000-0002-8959-4428}{Research supported in part by NSERC Discovery grant.}

\author{Danny Krizanc}{Wesleyan University, Middletown CT, USA \and \url{http://dkrizanc.web.wesleyan.edu/} }{dkrizanc@wesleyan.edu}{}{}

\author{Oscar Morales-Ponce}{California State University, Long Beach, CA, USA \and \url{https://home.csulb.edu/~omorales/} }{Oscar.MoralesPonce@csulb.edu}{https://orcid.org/0000-0002-9645-1257}{}

\authorrunning{J. Coleman and E. Kranakis and D. Krizanc and O. Morales-Ponce} 
\Copyright{Jared Coleman and Evangelos Kranakis and Danny Krizanc and Oscar Morales Ponce} 

\ccsdesc[300]{Theory of computation~Online algorithms}
\ccsdesc[300]{Theory of computation~Adversary models}

\keywords{Infinite Line, Knowledge, Oblivious, Robot, Search, Search-Time, Speed, Target} 

\nolinenumbers
\hideLIPIcs 

\EventEditors{John Q. Open and Joan R. Access}
\EventNoEds{2}
\EventLongTitle{42nd Conference on Very Important Topics (CVIT 2016)}
\EventShortTitle{CVIT 2016}
\EventAcronym{CVIT}
\EventYear{2016}
\EventDate{December 24--27, 2016}
\EventLocation{Little Whinging, United Kingdom}
\EventLogo{}
\SeriesVolume{42}
\ArticleNo{23}

\begin{document}

\setcounter{tocdepth}{10}

\maketitle

\begin{abstract}
Consider search on an infinite line involving an autonomous robot starting at the origin of the line and an oblivious moving target at initial distance $d \geq 1$ from it. The robot can change direction and move anywhere on the line with constant maximum speed $1$ while the target is also moving on the line with constant speed $v>0$ but is unable to change its speed or direction. The goal is for the robot to catch up to the target in as little time as possible.

The classic case where $v=0$ and the target's initial distance $d$ is unknown to the robot is the well-studied ``cow-path problem''. Alpert and Gal~\cite{alpern2003theory} gave an optimal algorithm for the case where a target with unknown initial distance $d$ is moving {\em away} from the robot with a known speed $v<1$. In this paper we design and analyze search algorithms for the remaining possible knowledge situations, namely, when $d$ and $v$ are known, when $v$ is known but $d$ is unknown, when $d$ is known but $v$ is unknown, and when both $v$ and $d$ are unknown. Furthermore, for each of these knowledge models we consider separately the case where the target is moving away from the origin and the case where it is moving toward the origin. We design algorithms and analyze competitive ratios for all eight cases above. The resulting competitive ratios are shown to be optimal when the target is moving towards the origin as well as when $v$ is known and the target is moving away from the origin.
\end{abstract}

\newpage
\section{Introduction}\label{sec:intro}
Search is important to many areas of computer science and mathematics and has received the attention of numerous studies. 
In the simplest search scenario, one is interested in the optimal trajectory of a single autonomous mobile agent (also referred to simply as a robot) tasked with finding a target placed at an unknown location on the infinite line. The line search problem is to give an algorithm for the agent so as to minimize the competitive ratio defined as the supremum over all possible target locations of the ratio of the time the agent takes to find the target and the time it would take if the target's initial position was known to the robot ahead of time. This classic problem has led to many variations (see~\cite{alpern2003theory} for more on its history).

In this paper we consider an extension of the line search problem involving an autonomous robot and an oblivious moving target. The search is again performed on an infinite line and concerns an autonomous robot starting at the origin of the line but differs from the previously studied case in that the search is for a {\em moving} target whose speed and direction are not necessarily known to the searching robot. 
The robot starts at the origin and the target at an arbitrary distance $d$ from the origin. 
The target is moving with constant speed and is oblivious in that it cannot change its speed and/or direction of movement. 
We consider and analyze several alternative knowledge-based scenarios in which the target's speed and initial distance from the origin may be known or unknown to the searching robot. 
The case where a target with unknown initial distance from the origin is moving away from the origin was solved by Alpern and Gal~\cite{alpern2003theory}. 
As far as we are aware, these are the first results for the remaining cases. 

\subsection{Notation and terminology}\label{sec:notation}
On the infinite real line, consider an autonomous robot which is initially placed at the origin whose maximum speed is $1$ and an oblivious robot (also referred to as the moving target) initially placed at a distance $d$ to the right or left of the origin and moving with constant speed $v > 0$. 
As is usually done in linear search and in order to avoid trivial considerations on the competitive ratio by adversarially placing the target very close to the robot, we assume that $d$ is not smaller than the unit distance, i.e., $d \geq 1$. 

The target may be moving away from or toward the origin. If it is moving away, we assume its speed is strictly less than $1$ as otherwise the problem can not be solved. 
Further, we assume that the autonomous robot knows the direction the target is moving (away from or toward the origin).
The search is completed as soon as the robot and target are co-located. 

The movement of the autonomous robot is determined by a trajectory which is defined as a continuous function $t \to f(t)$, with $f(t)$ denoting the location of the robot at time $t$. 
Moreover, it is true that $|f(t) - f(t')| \leq u |t-t'|$, for all $t, t'$, where $u$ is the speed of the agent (be that the searching robot or the oblivious target). The autonomous robot can move with its own constant speed and during the traversal of its trajectory it may stop and/or change direction instantaneously and at any time as specified by the search algorithm.

A search strategy is a sequence of movements followed by the robot. The competitive ratio of a search strategy $X$, denoted $CR_X$, is defined as the supremum over all possible initial target locations and speeds of the ratio of the time the agent takes to find the target and the time it would take if the target's initial position was known to the robot ahead of time. The competitive ratio of a certain type of search problem is the infimum of $CR_X$ taken over all possible strategies $X$ for this problem. By abuse of notation we may drop mention of $X$ when this is easily implied from the context.

Our goal in this paper is to prove bounds on the competitive ratios of algorithms under four different knowledge models:
\begin{enumerate}
  \item {\tt FullKnowledge}: The robot knows both the target's speed $v$ and its initial distance $d$.
  \item {\tt NoDistance}: The robot knows the target's speed $v$ but not its initial distance $d$.
  \item {\tt NoSpeed}: The robot knows the target's initial distance $d$, but not its speed $v$.
  \item {\tt NoKnowledge}: The robot knows neither the target's speed $v$ nor its initial distance $d$.
\end{enumerate}
For all knowledge models, the robot does not know the target's initial position.
We study each of the above knowledge models for the case where the target is moving toward the origin ({\tt Toward}) and where it is moving away ({\tt Away}) from the origin.
In each case, we assume the robot knows the direction of travel of the target. 

\subsection{Related Work}\label{sec:related}

Several research papers have considered the search problem for a robot searching for a static (fixed) target placed at an unknown location on the real line, see~\cite{BCR93,schuierer2001lower}. The problem was first independently considered in a stochastic setting by Bellman and Beck in the 1960's (cf. \cite{beck1964linear,bellman1963optimal} as well as \cite{BCR93,schuierer2001lower}). In a deterministic setting it is now well known that the optimal trajectory for this single agent search uses a doubling strategy whose trajectory attains a competitive ratio of $9$. Linear search has attracted much attention and been the focus of books including~\cite{ahlswede1987search,alpern2003theory,stone1975theory}.

The case of a moving target appears to have been first considered by McCabe \cite{mccabe1974}. In that paper, the problem of searching for an oblivious target that follows a Bernoulli random walk on the integers is considered. For the case of a deterministic oblivious searcher, the only result we are aware of us is found in Alpern and Gal \cite{alpern2003theory}. There 
they consider the case where the target is moving away from the origin at a constant speed $v<1$ which is known to the searching robot. Only the initial distance of the target is unknown. They give an algorithm with optimal competitive ratio for this case. 

Our problem is reminiscent of the problem of catching a fugitive in a given domain which is generally referred to as the cops and robbers problem \cite{anthony2011game}.
The main difference is that in those problems, the target (robber) is itself an autonomous agent. 
As a result, the techniques considered there do not apply to our case. 

Our problem is  also related to rendezvous (of two robots) on an infinite line but it differs because in our case only one of the robots is autonomous while the other is oblivious. Related studies on the infinite line include rendezvous with asymmetric clocks \cite{czyzowicz2018linear}  and asynchrnous deterministic rendezvous \cite{de2006asynchronous}. More recent work on linear search concerns searching for a static target by a group of cooperating robots, some of which may have suffered either crash \cite{czyzowicz2019searchcrash} or Byzantine \cite{czyzowicz2021searchbyz} faults. 

\subsection{Results of the paper}\label{sec:results}

In all situations considered it is unknown to the robot whether the target is initially to the left or to the right of the origin. 
We analyze the competitive ratio in four situations which reflect what knowledge the robot has about the target.
We present results on the {\tt FullKnowledge} model (the robot knows $v$ and $d$) in Section~\ref{sec:full_knowledge}, the {\tt NoDistance} model (the robot knows $v$ but not $d$) in Section~\ref{sec:no_dist}, the {\tt NoSpeed} model (the robot knows $d$ but not $v$) in Section~\ref{sec:no_speed}, and the {\tt NoKnowledge} model (the robot knows neither $v$ nor $d$) in Section~\ref{sec:no_knowledge}.
For each of these models we study separately the case when the target is moving away or toward the origin (this knowledge being available to the robot).
The results are summarized in Table~\ref{tab:results}.
We conclude with a summary and additional open problems.
\begin{table}[]
    \centering
    \begin{tabular}{| c | l | l | c |}
    \hline
    \textbf{Knowledge}                   & \textbf{Movement} & \textbf{Competitive Ratio}     & \textbf{Section}                        \\ \hline
    \multirow{2}{*}{$v,d$} & {\tt Away}        & $CR=1+ \frac{2}{1-v}$ & \ref{sec:full_knowledge_away}   \\ \cline{2-4} 
                                         & {\tt Toward}      & 
      \begin{tabular}[c]{@{}l@{}}$CR = 1+\frac{2}{1+v}$ \hspace{5.7em}if  $v < 1$ \\ $CR=1+\frac{1}{v}$ \hspace{6.6em}otherwise\end{tabular}   & \ref{sec:full_knowledge_toward} \\ \hline
    \multirow{2}{*}{$v$}    & {\tt Away}        & $CR = 1+8\frac{1+v}{(1-v)^2}$ & \ref{sec:no_dist_away}~\cite{alpern2003theory}      \\ \cline{2-4} 
     &
      {\tt Toward} &
      \begin{tabular}[c]{@{}l@{}}$CR = 1+\frac{1}{v}$ \hspace{6.6em}if  $v \geq \frac{1}{3}$ \\ $CR = 1+8\frac{1-v}{(1+v)^2}$ \hspace{4em}otherwise\end{tabular} &
      \ref{sec:no_dist_toward} \\ \hline
    \multirow{2}{*}{$d$} &
      {\tt Away} &
      \begin{tabular}[c]{@{}l@{}}$CR \leq 5$  \hspace{8.45em}if $v \leq \frac{1}{2}$\\ $CR \leq 1+16\frac{\left(\log\frac{1}{1-v}\right)^2}{(1-v)^4}$ \hspace{2em}otherwise\end{tabular} &
      \ref{sec:no_speed_away} \\ \cline{2-4} 
                                         & {\tt Toward}      & $CR = 3$                       & \ref{sec:no_speed_toward}       \\ \hline
    \multirow{2}{*}{$\emptyset$} &
      {\tt Away} &
      \begin{tabular}[c]{@{}l@{}}$CR \leq 1 + \frac{16}{d} \left[ \log \log \left( \max\left(d, \frac{1}{1-v}\right)\right) + 3\right]$\\ \hspace{4.5em}$\cdot \max\left(d, \frac{1}{1-v}\right)^{8} \cdot \log^2\left[ \max\left(d, \frac{1}{1-v}\right) \right]$\end{tabular} &
      \ref{sec:no_knowledge_away} \\ \cline{2-4} 
                                         & {\tt Toward}      & $CR = 1+\frac{1}{v}$           & \ref{sec:no_knowledge_toward}   \\ \hline
    \end{tabular}
    \caption{Table of competitive ratio bounds proven for each knowledge model for cases with the target moving away from or towards the origin with speed $v$ and initial distance $d$ from the robot which is moving with speed $1$. Equalities indicate that tight upper and lower bounds are proven.}
    \label{tab:results}
\end{table}

\section{The {\tt FullKnowledge} Model}
\label{sec:full_knowledge}

We first study the model where the robot knows the target's speed $v$ and its initial distance from the origin $d$.

\subsection{The {\tt FullKnowledge/Away} Model}
\label{sec:full_knowledge_away}

For the case when the target is moving away from the origin, clearly if $v \geq 1$ then the robot can never catch the target.
Thus, for this model (and all other {\tt Away} models), we assume $v < 1$.
In this section, we will analyze an algorithm where the robot chooses a direction and moves for time $\frac{d}{1-v}$. 
If the robot does not find the target after moving for time $\frac{d}{1-v}$ in one direction, then it changes direction and continues moving until it does.

\begin{algorithm}[H]
  \caption{Online Algorithm for {\tt FullKnowledge/Away} Model}\label{alg:full_knowledge_away}
  \begin{algorithmic}[1]
    \State {{\bf input:} target speed $v$ and initial distance $d$}
    \State{choose any direction and go for time $\frac d{1-v}$}
    \If{target not found}
      \State{change direction and go until target is found}
    \EndIf
  \end{algorithmic}
\end{algorithm}

\begin{theorem}\label{thm:full_knowledge_away}
  For the {\tt FullKnowledge/Away} model, Algorithm~\ref{alg:full_knowledge_away} has an optimal competitive ratio of
  \begin{equation}
    \label{eq:full_knowledge_away}
    1 + \frac{2}{1-v}.
  \end{equation}
\end{theorem}
\begin{proof}
  By Algorithm~\ref{alg:full_knowledge_away}, the robot goes in one direction for a time $\frac d{1-v}$.
  Observe that if the robot does not encounter the target after this amount of time, it must be on the opposite side of the origin (in the other direction).
  At the time the robot changes direction, its distance to the target will be equal to $\frac d{1-v} + d + \frac {dv}{1-v} = \frac{2d}{1-v}$. 
  Thus, the total time required until the robot catches up to the target is at most
  $$\frac{d}{1-v} +  \frac{2d}{(1-v)^2}.$$
  Clearly then, the competitive ratio is at most
  $$
    \frac{\frac{d}{1-v} +  \frac{2d}{(1-v)^2}}{\frac d{1-v}} = 1 + \frac{2}{1-v}
  $$
  which is as claimed in Equation~\eqref{eq:full_knowledge_away} above.

  Optimality follows from the fact that regardless of which direction the robot chooses to travel, the adversary can place the target in the opposite direction. 
  Moreover, for the robot to catch up to the target it must visit one of the points $\pm \frac d{1-v}$. 
  If the robot visits location $\frac d{1-v}$ to the right (resp. left) the adversary places the target on the left (resp. right). 
  Therefore the completion time will be at least $\frac d{1-v} + \frac{2d}{(1-v)^2} $. 
  This shows the upper bound is tight and completes the proof of Theorem~\ref{thm:full_knowledge_away}. 
\end{proof}

\subsection{The {\tt FullKnowledge/Toward} Model}
\label{sec:full_knowledge_toward}

Consider the following algorithm which is similar to Algorithm~\ref{alg:full_knowledge_away}.
\begin{algorithm}[H]
  \caption{Online Algorithm for the {\tt FullKnowledge/Toward} Model}\label{alg:full_knowledge_toward}
  \begin{algorithmic}[1]
    \State {{\bf input:} target speed $v$ and initial distance $d$}
    \State{choose any direction and go for time $\frac{d}{1+v}$}
    \If{target not found}
      \State{change direction and go until target is found}
    \EndIf
  \end{algorithmic}
\end{algorithm}

\begin{theorem}\label{thm:full_knowledge_toward_upper}
  For the {\tt FullKnowledge/Toward} model, Algorithm~\ref{alg:full_knowledge_toward} has competitive ratio at most
  \begin{equation}\label{eq:full_knowledge_toward_upper}
    1 + \frac{2}{1+v}.
  \end{equation}
\end{theorem}
\begin{proof}
  The robot goes in one direction for a time $\frac d{1+v}$.
  If the robot finds the target in this time, the algorithm is clearly optimal.
  If, however, the robot does not find the target, then it must be on the opposite side of the origin (in the other direction).
  If this is the case, then by time $\frac{d}{1+v}$ the target has moved a distance $\frac{dv}{1+v}$ and is at distance $d - \frac{dv}{1+v} = \frac d{1+v}$ from the origin.
  Therefore at the time the robot changes direction, the distance between robot and target is $\frac {2d}{1+v}$.
  Thus, the robot will encounter the target in additional time $\frac {2d}{(1+v)^2}$. 
  It follows that the total time required for the robot to meet the target is $\frac d{1+v} + \frac {2d}{(1+v)^2}$ and the resulting competitive ratio satisfies
  $$
    CR \leq \frac{\frac d{1+v} + \frac {2d}{(1+v)^2}}{\frac d{v+1}} 
    = 1+ \frac 2{1+v}.
  $$
  This completes the proof of Theorem~\ref{thm:full_knowledge_toward_upper}. 
\end{proof}

\begin{theorem}\label{thm:full_knowledge_toward_lower}
  For the {\tt FullKnowledge/Toward} model, the competitive ratio of any online algorithm is at least $1 + \frac 2{1+v}$, provided that $v<1$.
  In particular, Algorithm~\ref{alg:full_knowledge_toward} is optimal for $v<1$.
\end{theorem}
\begin{proof}
  Consider any algorithm for a robot starting at the origin to meet a target initially placed at an unknown location distance $d$ from the origin.
  For any point at distance $a$ from the origin, the target takes exactly $\frac{d-a}{v}$ time to reach $a$.
  Then, let $t$ denote the time the robot first passes through a point at distance $a$ from the origin. 
  If $t < \frac{d-a}{v}$, then the robot cannot know whether the target is on the same or opposite side of the origin.
  On the other hand, if $t \geq \frac{d-a}{v}$ and it has not encountered the target, then the target must be on the opposite side of the origin.
  Thus, given a trajectory, let $\pm a$ be the first point such that the robot is at position $\pm a$ at time exactly $\frac{d-a}{v}$.
  Clearly such a point must exist for any trajectory since the target is moving toward the origin.
  Then whichever side of the origin the robot is on, consider the instance where the target started on the opposite side.
  Clearly then, the robot takes an additional time at least $\frac{2a}{1+v}$ to reach the target.
  Thus, the competitive ratio is given by:
  \begin{align}
    \frac{\frac{d-a}{v} + \frac{2a}{1+v} }{\frac d{1+v}} 
    &= \label{eq:full_knowledge_toward_lower}
    \frac{(d-a)(1+v) + 2av}{vd} =
    1 + \frac1v + \frac{a(v-1)}{dv} .
  \end{align}
  Observe that whenever $v<1$, the right-hand side of Equation~\eqref{eq:full_knowledge_toward_lower} satisfies
  $$
    1 + \frac1v + \frac{a(v-1)}{dv} \geq 
    1 + \frac1v + \frac{\frac d{1+v} (v-1)}{dv} =
    1 + \frac 2{1+v}
  $$
  which completes the proof of Theorem~\ref{thm:full_knowledge_toward_lower}.
\end{proof}

With Theorem~\ref{thm:full_knowledge_toward_lower} proved, we know Algorithm~\ref{alg:full_knowledge_toward} is optimal for any value of $v$ between $0$ and $1$, but what about when $v>1$? In this case, we'll prove the following algorithm is optimal: the robot waits at the origin forever. We call this algorithm ``the waiting algorithm''. 

\begin{theorem}\label{thm:waiting}
  Whenever the target is moving toward the origin with speed $v \geq 1$, the waiting algorithm has an optimal competitive ratio of $1 + \frac{1}{v}$.
\end{theorem}
\begin{proof}
  Clearly the algorithm takes exactly time $d/v$ to complete and so the upper bound follows trivially.
  For the lower bound, we build upon the proof of Theorem~\ref{thm:full_knowledge_toward_lower}.
  It simply remains to consider Equation~\eqref{eq:full_knowledge_toward_lower} for $v > 1$.
  In this case, the right-hand side of Equation~\eqref{eq:full_knowledge_toward_lower} is increasing with respect to $a \geq 0$, so
  $$
    1 + \frac1v + \frac{a(v-1)}{dv} \geq 1 + \frac 1v.
  $$
  This completes the proof of Theorem~\ref{thm:waiting}.
  
\end{proof}

\begin{remark}\label{rmk:waiting}
  Observe that the waiting algorithm makes no use of the target's speed or initial distance and therefore, as long as the target is moving toward the origin, applies directly to the other knowledge models.
\end{remark}

\section{The {\tt NoDistance} Model}\label{sec:no_dist}
In this section we assume that the robot knows $v$ but not $d$.
Consider the following zig-zag algorithm with ``expansion ratio'' $a > 0$ (with the value of $a$ to be determined).

\begin{algorithm}[H]
  \caption{Online Algorithm for {\tt NoDistance/Away} and {\tt NoDistance/Toward} Models}\label{alg:no_dist}
  \begin{algorithmic}[1]
    \State {\bf input}: target speed $v$ and expansion ratio $a$
    \State $i \gets 0$
    \While {target not found}
      \If {at origin}
        \State $d \gets (-a)^i$
        \State $i \gets i + 1$
      \ElsIf {at $d$}
        \State $d \gets 0$
      \EndIf
      \State {move toward $d$}
    \EndWhile
  \end{algorithmic}
\end{algorithm}

\subsection{The {\tt NoDistance/Away} Model}\label{sec:no_dist_away}
The following result was shown by Alpern and Gal~\cite{alpern2003theory}.
\begin{theorem}\label{thm:no_dist_away_upper}
  For the {\tt NoDistance/Away} model, Algorithm~\ref{alg:no_dist} with $a = 2 \frac{1+v}{1-v}$ has an optimal competitive ratio of
  \begin{align*} 
      1+ 8\frac{1+v}{(1-v)^2}.
  \end{align*}
\end{theorem}

\subsection{The \tt {NoDistance/Toward} Model}
\label{sec:no_dist_toward}

Recall first the statement made in Remark~\ref{rmk:waiting}, that the optimality of the waiting algorithm (which makes no use of any knowledge of $d$) holds for any $d$ as long as $v \geq 1$.
Thus, we need only consider scenarios where $0 \leq v < 1$.
As we will see, however, when the target is moving toward the origin, the waiting algorithm is optimal for far slower targets!
In general, since the target is moving toward the origin, the robot need not search ever-increasing distances away from the origin (i.e. execute Algorithm~\ref{alg:no_dist} with an expansion ratio $a > 1$).
We call any algorithm which involves the robot {\em never} traveling further than some finite distance from the origin (in one or both directions) a {\em contracting algorithm}. 
Note that Algorithm~\ref{alg:no_dist} for $0 < a \leq 1$ is a contracting algorithm and $a=0$ is exactly the waiting algorithm.
We'll start by showing that any contracting algorithm cannot have a better competitive ratio than the waiting algorithm:

\begin{theorem}\label{thm:contracting}
  The competitive ratio of Algorithm~\ref{alg:no_dist} for any $0 \leq a \leq 1$ is $1 + \frac{1}{v}$.
\end{theorem}
\begin{proof}
  Let $d^\prime$ be the finite distance further than which the robot will never travel in at least one direction.
  Then consider the scenario where the target is initially a distance $d = c \cdot d^\prime >> d^\prime$ from the origin in the same direction.
  Then the competitive ratio is at least 
  \begin{align*}
    \sup_{c} \frac{\frac{cd^\prime - d^\prime}{v}}{\frac{c d^\prime}{v+1}} &= \sup_{c} \frac{c-1}{c} \frac{1+v}{v}
    = \lim_{c \rightarrow \infty} \frac{c-1}{c} \left(1 + \frac{1}{v} \right)
    = 1 + \frac{1}{v}
  \end{align*}
  which proves Theorem~\ref{thm:contracting}.
\end{proof}

By Theorem~\ref{thm:contracting}, any algorithm which hopes to out-perform the waiting algorithm must be expanding.
Now we show that the following hybrid algorithm, Algorithm~\ref{alg:no_dist_toward}, is optimal.

\begin{algorithm}[H]
  \caption{Wait or Zig-Zag Search Algorithm for {\tt NoDistance}/{\tt Toward} model}\label{alg:no_dist_toward}
  \begin{algorithmic}[1]
    \State {{\bf input:} target speed $v$}
    \If{$v \geq \frac{1}{3}$}
      \State execute waiting algorithm
    \Else
      \State execute Algorithm~\ref{alg:no_dist} with $a=2\frac{1-v}{1+v}$
    \EndIf
  \end{algorithmic}
\end{algorithm}

\begin{theorem}\label{thm:no_dist_toward_upper}
  For the {\tt NoDistance/Toward} model, the competitive ratio of Algorithm~\ref{alg:no_dist_toward} is at most 
  \begin{align}
    \begin{cases} 
      1 + \frac{1}{v} & \mbox{if $v \geq \frac{1}{3}$} \\
      1 + 8\frac{1-v}{(1+v)^2} & \mbox{if $v < \frac{1}{3}$}
    \end{cases}
  \end{align}
\end{theorem}

\begin{proof}
  The first case is trivial: the competitive ratio of the waiting algorithm is exactly $1 + \frac{1}{v}$ by Theorem~\ref{thm:waiting}.
  The second case, however, is a bit more complicated.
  First, observe that if $v < \frac{1}{3}$ then $a$ must be less than $3$.
  Indeed, consider the scenario where the robot ``just misses'' the target on the very first round of the algorithm (after traveling a distance $1$ in some direction and then turning around).
  Then the competitive ratio of the algorithm is 
  \begin{align*}
    1+\frac{2a+2}{1+v}
  \end{align*}
  which is greater than $1+8\frac{1-v}{(1+v)^2}$ for any $a > 3$ and $v > 0$:
  \begin{align*}
    1+\frac{2a+2}{1+v} &> 1 + \frac{8}{1+v} > 1 + \frac{8}{1+v} \cdot \frac{1-v}{1+v}
  \end{align*}
  since $\frac{1-v}{1+v} < 1$.

  Now, consider the round $k$ when the robot catches up to the target and observe that
  \begin{align*}
    a^{k-2} &< d - \left(2 \sum_{i=0}^{k-3} a^i + a^{k-2}\right) v 
    = d - \left( 2 \frac{a^{k-2} - 1}{a-1} + a^{k-2} \right) v
  \end{align*}
  since otherwise, the robot would have caught up to the target in round $k-2$.
  This yields the following inequality which will prove useful in analyzing the competitive ratio below:
  \begin{align}
    a^{k-2} &< d - \left( 2 \frac{a^{k-2}}{a-1} + a^{k-2} \right) v \nonumber \\
    &\leq d \frac{a-1}{a-1+v(a+1)} + \frac{v}{a-1+v(a+1)} \nonumber \\
    &\leq d \frac{a-1}{a-1+v(a+1)} + \frac{1}{4a-2} \label{eq:no_dist_toward_upper:just_missed_bound}
  \end{align}
  
  Observe the worst competitive ratio, then, is given by the situation where the robot ``just misses'' the target on the $(k-2)^\text{th}$ round and catches up to it only on round $k$.
  It follows the competitive ratio of Algorithm~\ref{alg:no_dist_toward} is
  \begin{align*}
    \frac{
      2 \sum_{i=0}^{k-3} a^i + a^{k-2} + \frac{2(a^{k-2} + a^{k-1})}{1+v}
    }{
      \frac{d}{1-v}
    } \leq \frac{
      2 \frac{a^{k-2}-1}{a-1} + a^{k-2} + \frac{2(a^{k-2} + a^{k-1})}{1+v}
    }{
      \frac{d}{1-v}
    }
  \end{align*} 
  which, by Inequality~\eqref{eq:no_dist_toward_upper:just_missed_bound} (and by substituting each $a^{k-2}$ with the right-hand side of Inequality~\eqref{eq:no_dist_toward_upper:just_missed_bound}), is less than or equal to
  \begin{align}
    CR &\leq 1 + \frac{1}{2} \left[ \frac{1}{d} \left( \frac{a-3}{a-1} \cdot \frac{v(5-7a)}{1-3a+2a^2} \right) + \frac{4a^2}{(a-1)+v(a+1)} \right] \label{eq:no_dist_toward_upper:upper_bound_1} \\
    &\leq \lim_{d \rightarrow \infty} 1 + \frac{1}{2} \left[ \frac{1}{d} \left( \frac{a-3}{a-1} \cdot \frac{v(5-7a)}{1-3a+2a^2} \right) + \frac{4a^2}{(a-1)+v(a+1)} \right] \nonumber \\
    &= 1 + \frac{2a^2}{(a-1)+v(a+1)} \label{eq:no_dist_toward_upper:upper_bound_2}
  \end{align}
  which follows since the right-hand side of Inequality~\eqref{eq:no_dist_toward_upper:upper_bound_1} is increasing with respect to $d$ on $1 < a \leq 3$.
  Finally, the right-hand side of Inequality~\eqref{eq:no_dist_toward_upper:upper_bound_2} is minimized at $a=2 \frac{v-1}{v+1}$ with a value of $1+8\frac{1-v}{(1+v)^2}$, which proves Theorem~\ref{thm:no_dist_toward_upper}.
\end{proof}

Now we show that Algorithm~\ref{alg:no_dist_toward} is optimal by proving a tight lower bound on the competitive ratio for any online algorithm.
Our proof is based on techniques developed in~\cite{killick2022cone}.
Let $X(t)$ denote be the robot's position at time $t$ according to a given strategy.

\begin{theorem}\label{thm:no_dist_toward_lower}
    For the {\tt NoDistance/Toward} model, any strategy $X$ has a competitive ratio of at least
    \begin{align*}
        \begin{cases}
            1+\frac{1}{v} & \mbox{if  $v \geq \frac{1}{3}$} \\
            1 + 8\frac{1-v}{(1+v)^2} & \mbox{otherwise}
        \end{cases}
    \end{align*}
\end{theorem}

\begin{proof}
    Let $\beta_t = \inf_{t^\prime > t} \frac{t^\prime}{|X(t^\prime)|}$.
    Clearly, then, if $t_1 \leq t_2$ then $\beta_{t_1} \leq \beta_{t_2}$.
    Furthermore, $\beta_t \geq 1$ for all $t$ since the maximum speed of the robot is $1$.
    Now let $\beta = \lim_{t \rightarrow \infty} \beta_t$.
    By definition of the limit infimum, there must exist a finite time $t$ such that $\beta \leq \frac{t^\prime}{|X(t^\prime)|}$ for all $t^\prime \geq t$ and thus there must exist a time $t_1 > \frac{\beta (\beta + 1)}{\beta - 1} t$ such that the robot reaches a point (without loss of generality, on the right side of the origin) $X(t_1) = \frac{t_1}{\beta + \epsilon_1}$ for any arbitrarily small $\epsilon_1 > 0$.
    Consider such a time and observe that, by construction, the robot could not have reached any point to the left of $x_0 = -\frac{t_1-X(t_1)}{1+\beta}$ after time $t_0 = \frac{\beta(t_1-X(t_1))}{1+\beta}$ since $x_0 \leq -t$ and $t_0 > t$ (see Figure~\ref{fig:no_dist_toward_cone}).
    \begin{figure}[ht!]
      \centering
      \includegraphics[width=0.6\textwidth]{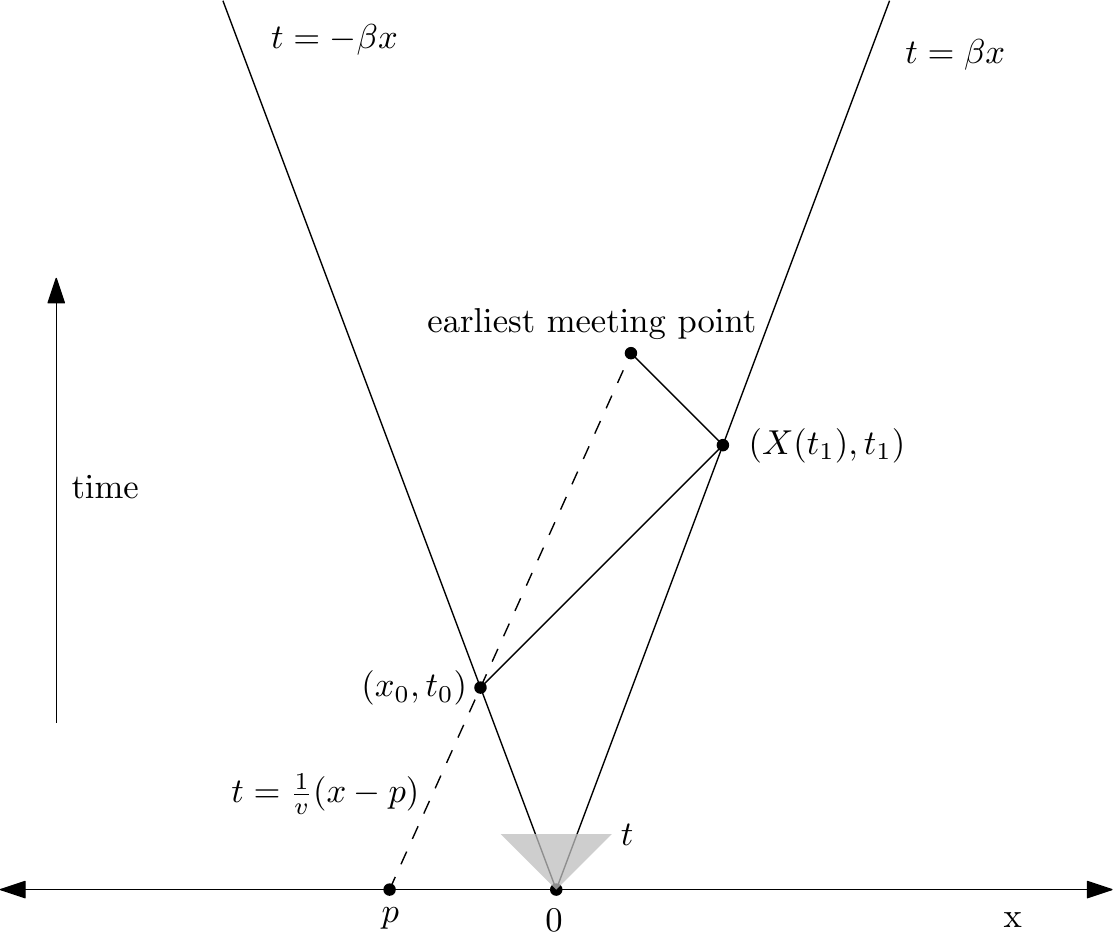}
      \caption{The cone-bounded trajectory of the robot and worst-case placement $p$ of the target. The small gray triangle is to remind the reader that, by the definition of $\beta$, the robot trajectory is only guaranteed to be contained by the cone after some finite time $t$. Thus, in order to maximize the competitive ratio, we (as the adversary) should place the target so that its trajectory does not intersect $(x_0, t_0)$ {\em or} the gray triangle.}
      \label{fig:no_dist_toward_cone}
    \end{figure}
    Now, consider a target starting at initial positon $p$ (to be determined) moving at speed $v>0$ toward a robot which starts at the origin and has a speed of $1$.
    Thus, by placing the target at a starting location so that the farthest {\em right} the robot could have reached is $x_0 - \epsilon_0$ for any arbitrarily small $\epsilon_0$, the robot {\em can not} have reached the target by time $t_1$.
    Such a target has an initial position of 
    \begin{align*}
        p = -\frac{(1+\beta v)(t_1-X(t_1))}{1+\beta} - \epsilon_0
    \end{align*}
    and follows the trajectory
    \begin{align}
        X_{\text{target}}(t) = v t + p\label{eq:no_dist_toward_traj_target}
    \end{align}
    where $X_{\text{target}}(t)$ denotes the robot's position at time $t$.
    Observe also, if the robot moves directly toward the target after $t_1$, then its trajectory after time $t_1$ is given by
    \begin{align}
        X(t) = X(t_1) + t_1 - t \label{eq:no_dist_toward_traj_robot}
    \end{align}
    Thus, the earliest time the robot could possibly encounter the target can be computed by finding the intersection between the robot trajectory (Equation~\eqref{eq:no_dist_toward_traj_robot}) and the target's trajectory (Equation~\eqref{eq:no_dist_toward_traj_target}) and solving for $t$:
    \begin{align}
      v t + p &= X(t_1) + t_1 - t \nonumber \\
      t &= \frac{X(t_1)+t_1-p}{1+v} . \label{eq:no_dist_toward_meeting_time}
    \end{align}

    Then the competitive ratio (Equation~\eqref{eq:no_dist_toward_meeting_time} divided by $-p/(1+v)$, the optimal search time) is
    \begin{align}
        CR &\geq \sup_{\epsilon_0,\epsilon_1} \frac{(X(t_1)+t_1-p)/(1+v)}{-p/(1+v)} 
          = \sup_{\epsilon_0,\epsilon_1} \frac{X(t_1)+t_1-p}{-p} =  \sup_{\epsilon_0,\epsilon_1} \left[ 1 - \frac{X(t_1)+t_1}{p} \right] \nonumber \\
        &= \sup_{\epsilon_1} \left[ 1 + \frac{(1+\beta)(t_1+X(t_1))}{(1+\beta v)(t_1-X(t_1))} \right]
          \label{eq:no_dist_toward_lower_1} \\
        &= 1 + \frac{(1+\beta)^2}{(1 + \beta v)(\beta - 1)}
          \label{eq:no_dist_toward_lower_2}
    \end{align}
    where Inequality~\eqref{eq:no_dist_toward_lower_1} follows since $p = -\frac{(1+\beta v)(t_1-X(t_1))}{1+\beta} - \epsilon_0$ for arbitrarily small $\epsilon_0>0$ and Inequality~\eqref{eq:no_dist_toward_lower_2} follows since $X(t_1) = \frac{t_1}{\beta+\epsilon_1}$ for arbitrarily small $\epsilon_1>0$.
    Finally, observe that if $v < \frac{1}{3}$, then the right-hand side of Equality~\eqref{eq:no_dist_toward_lower_2} has a single minimum of $1 + 8 \frac{1-v}{(1+v)^2}$ at $\beta = \frac{v-3}{3v-1}$.
    On the other hand, if $v \geq \frac{1}{3}$, then the right-hand side of Equality~\eqref{eq:no_dist_toward_lower_2} is decreasing with respect to $\beta$ and thus the competitive ratio satisfies
    \begin{align*}
        CR \geq \lim_{\beta \rightarrow \infty} \left[ 1 + \frac{(1+\beta)^2}{(1 + \beta v)(\beta - 1)} \right] = 1 + \frac{1}{v}.
    \end{align*}
\end{proof}

\section{The {\tt NoSpeed} Model}\label{sec:no_speed}

In this section we assume that the robot knows $d$ but not $v$.

\subsection{The {\tt NoSpeed/Away}}\label{sec:no_speed_away}

For this model, it is clear that the robot cannot execute an algorithm like Algorithm~\ref{alg:full_knowledge_away} since no upper bound on the target's speed is known.
Note that, if any upper bound $\hat{v} < 1$ on the target's speed {\em were} known, the robot could execute Algorithm~\ref{alg:full_knowledge_away} by assuming the target speed to be equal to $\hat{v}$, resulting in a competitive ratio of at most $1 + \frac{2}{1-\hat{v}}$.
Since the target speed $v$ is unknown (and potentially very close to $1$), however, we propose another strategy. 
Consider a monotone increasing non-negative integer sequence $\{ f_i : i \geq 0 \}$ such that $f_0 = 1$ and $f_i < f_{i+1}$, for all $i\geq 0$.
The idea of the algorithm is to search for the target by making a guess about its speed in rounds as follows. 
We start from the origin and alternate searching right and left. 
On the $i$-th round, we use the guess $v_i = 1 - 2^{-f_i}$ and search the necessary distance away from the origin such that, if the target's speed is less than or equal to $v_i$ and the target's initial position is in the same direction from the origin that the robot moves in round $i$, then the target will be found in round $i$.
Otherwise, we can conclude that either the target is moving with a speed greater than $v_i$ or else it is on the opposite side of the origin.
In this case the robot returns to the origin and repeats the algorithm in the opposite direction.

Later in the analysis we will show how to select the integer sequence $\{ f_i : i \geq 0\}$ so as to obtain bounds on the competitive ratio. 
The algorithm explained above is formalized as Algorithm~\ref{alg:no_speed_away}.  

\begin{algorithm}[H]
  \caption{Online Algorithm for {\tt NoSpeed/Away} Model}\label{alg:no_speed_away}
  \begin{algorithmic}[1]
    \State {{\bf input:} target initial distance $d$\\
    ~~~~~~~~~~~~integer sequence $\{ f_i : i \geq 0\}$ such that $f_i < f_{i+1}$, for $i \geq 0$ and $f_0 = 1$;}

    \State $t \gets 0$
    \For {$i \leftarrow 0, 1, 2, \ldots$ until target found}
      \State $v_i \gets 1-2^{-f_i}$
      \State $x_i \gets (-1)^i \cdot \frac{d + tv_i}{1-v_i}$
      \State move to $x_i$ and back to the origin
      \State $t \gets t + |x_i|$
    \EndFor
  \end{algorithmic}
\end{algorithm}

To compensate for the fact that the starting speed of the robot in the algorithm is $v_0 = 1 - 2^{-1}= 1/2$ we first need to consider the case $v \leq \frac 12$.

\begin{lemma}\label{lm:no_speed_upper_1}
  For the {\tt NoSpeed/Away} model, if the unknown speed $v$ of the target is less than or equal to $\frac{1}{2}$ then the competitive ratio of Algorithm~\ref{alg:no_speed_away} is at most $5$.
\end{lemma}
\begin{proof}
  According to Algorithm~\ref{alg:no_speed_away} and since $v < 1/2$, the robot will find the target either on its first trip away from the origin, after time at most $2d$, or after the first time it changes direction of movement. 
  In the worst case it will spend time $2d$ in one direction and then additional time $\frac{2d + d + 2dv}{1-v}$.
  It follows that the competitive ratio is at most
  \begin{align*}
    \frac{2d + \frac{2d+ d +2dv}{1-v} }{\frac d{1-v}} = 5
  \end{align*}
  which proves Lemma~\ref{lm:no_speed_upper_1}. 
\end{proof}

Next we analyze the competitive ratio of the algorithm for $v > \frac 12$. 
\begin{lemma}\label{lm:no_speed_upper_2}
  For the {\tt NoSpeed/Away} model, if the unknown speed $v$ of the target is greater than $\frac{1}{2}$ then the competitive ratio of Algorithm~\ref{alg:no_speed_away} is at most $1 + 2^{1+\sum_{j=0}^k f_{j}} \cdot 4^{k+1}$ where $k$ is the first $k$ such that $v_k \geq v$.
\end{lemma}
\begin{proof}
  Let $d_i$ be the distance from the origin the target would be if its speed was equal to $v_i$, where $v_i = 1 - 2^{-f_i}$ at time $\sum_{j=0}^{i-1} \left( 1 - 2^{-f_i} \right)$.
  In other words, if $v_i \geq v$, then $d_i$ is the maximum distance of the target from the origin (and thus, the robot) at the beginning of round $i$ of the algorithm.
  Thus, if the speed of the target is less than or equal to $v_i$ and the robot moves toward it in round $i$, then it would take at most $x_i = \frac{d_i}{1-v_i} = 2^{f_i} d_i$ additional time for the robot to catch up to the target, for $i \geq 0$.
  Recall the algorithm involves the robot moving a distance $x_i$ (in time $x_i$, since the robot's speed is $1$) away from the origin and back in round $i$.
  Observe then that $d_0 = d$, $v_0 = 1/2$, and
  \begin{align}
    d_i = d+ 2v_i \sum_{j=0}^{i-1}  x_j. \label{eq:no_speed_upper_guess_d}
  \end{align}
  Therefore, it follows from the definition of $x_i$ that
  \begin{align}
    x_i = 2^{f_i} \left( d+ 2v_i \sum_{j=0}^{i-1}  x_j \right) . \label{eq:no_speed_upper_guess_x}
  \end{align}
  As a consequence
  \begin{align}
    \sum_{j=0}^{i-1}  x_j = \frac{x_i - 2^{f_i} d}{2^{f_i} \cdot 2 \cdot v_i} \label{eq:no_speed_upper_guess_3}
  \end{align}
  Similarly, if we replace $i$ with $i+1$ we have that
  \begin{align}
    \sum_{j=0}^{i} x_j = \frac{x_{i+1} - 2^{f_{i+1}} d}{2^{f_{i+1}} \cdot 2 \cdot v_{i+1}} \label{eq:no_speed_upper_guess_4}
  \end{align}
  Subtracting Equation~\eqref{eq:no_speed_upper_guess_3} from  Equation~\eqref{eq:no_speed_upper_guess_4}, we derive the recurrence
  \begin{align}
    x_i = \frac{x_{i+1} - 2^{f_{i+1}} d}{2^{f_{i+1}} \cdot 2 \cdot v_{i+1}} - \frac{x_i - 2^{f_i} d}{2^{f_i} \cdot 2 \cdot v_i} \label{eq:no_speed_upper_guess_5} 
  \end{align}
  Collecting similar terms and simplifying Equation~\eqref{eq:no_speed_upper_guess_5} yields
  \begin{align}
    \frac{x_{i+1}}{2^{f_{i+1}} \cdot 2 \cdot v_{i+1}}
      &= \notag \left( 1 + \frac 1{2^{f_i} \cdot 2 \cdot v_i}\right) x_i
      + \left( \frac{2^{f_{i+1}}}{2^{f_{i+1}} \cdot 2 \cdot v_{i+1}} - 
      \frac{2^{f_i}}{2^{f_i} \cdot 2 \cdot v_{i}} \right) d \\
    &= x_i \left( 1+\frac{1}{2^f_i \cdot 2 \cdot v_i} \right) + \frac{d}{2 v_{i+1}} - \frac{d}{2 v_{i}} \label{eq:no_speed_upper_guess_5_inter} \\
    &\leq \left( 1 + \frac 1{2^{f_i} \cdot 2 \cdot v_i}\right) x_i \label{eq:no_speed_upper_guess_5_new} 
  \end{align}
  following from the fact the sum of the last two terms in Inequality~\ref{eq:no_speed_upper_guess_5_inter} is less than or equal to $0$.

  If we simplify the right-hand side of Equation~\eqref{eq:no_speed_upper_guess_5_new}, we derive the following recursive inequalities
  \begin{align}
    x_{i+1} &\leq 2^{f_{i+1}} \cdot 2 \cdot v_{i+1} \left( 1+ \frac 1{2^{f_i} \cdot 2 \cdot v_i}\right) x_i \notag \\
    &\leq 2^{f_{i+1}} \cdot 2 \cdot \left( 1+ \frac 1{2^{f_i}}\right) x_i \notag \\
    &\leq \left( 2 \cdot 2^{f_{i+1}} + 2\cdot 2^{f_{i+1} - f_i}\right) x_i \notag
    \\
    &\leq 2^{f_{i+1} } \cdot 4 \cdot x_i , \label{eq:no_speed_upper_guess_6} 
  \end{align}
  which follows since $\frac{1}{2} \leq v_i < 1$ and $f_i < f_{i+1}$ for all $i$.

  By repeated application of the last Recurrence~\eqref{eq:no_speed_upper_guess_6} above and using the fact that by definition $x_0 = 2^{f_0}d$, it follows easily by induction that
  \begin{align}
    x_{i+1} 
    &\leq \notag 2^{f_{i+1} } \cdot 4 \cdot x_i \\
    &\leq \notag 2^{f_{i+1}+f_i} \cdot 4^2 \cdot x_{i-1}  \\
    &~ \vdots \notag \\
    &\leq \notag 2^{f_{i+1}+f_i+f_{i-1}+ \cdots + f_1} \cdot 4^{i+1} \cdot x_0  \\
    &\leq \label{eq:no_speed_upper_guess_7} 2^{\sum_{j=0}^{i+1} f_j} \cdot 4^{i+1} \cdot d 
  \end{align}

  Consider the first $i$ such that $v_i \geq v$. 
  It follows that and $v_{i-1} < v$ which yields $1-v < 1-v_{i-1 } = 2^{-f_{i-1}} $ and implies that $2^{f_{i-1}} < \frac 1{1-v}$. 
  Note that although $v_i \geq v$, the robot may not find the target in round $i$ because it is located in the opposite direction.
  It is guaranteed, however, to find the target by round $i+1$.
  Moreover the total time that has elapsed from the start until round $i$ is $2 \sum_{j=0}^i x_j$ at which time the target is at distance $d + v 2 \sum_{j=0}^i x_j$ from the origin. 

  As a consequence the competitive ratio of Algorithm~\ref{alg:no_speed_away} is at most
  \begin{align*}
    \frac{2 \sum_{j=0}^i x_j + \frac{d+ 2 v \sum_{j=0}^i x_j}{1-v}}{\frac d{1-v}} &= 1 + \frac{2(v + 1-v)}d \sum_{j=0}^i x_j
      \\
    &= 1 + \frac{x_{i+1} - 2^{f_{i+1}}d}{d} \cdot \frac{2}{2^{f_{i+1}}\cdot 2 \cdot v_{i+1}}
      &&\text{(By~\eqref{eq:no_speed_upper_guess_4})} \\ 
    &\leq 1 + \frac{x_{i+1}}d \cdot \frac{2}{2^{f_{i+1}}\cdot 2 \cdot v_{i+1}}
      \\ 
    & \leq 1+ \frac{1}{v_{i+1}} 2^{\sum_{j=0}^i f_{j}} \cdot 4^{i+1} 
      &&\text{(By~\eqref{eq:no_speed_upper_guess_7})} 
  \end{align*}
  Since $v_{i+1} \geq 1/2$ we conclude with an upper bound on the competitive ratio of Algorithm~\ref{alg:no_speed_away} of
  \begin{align}
    1+ 2^{1+\sum_{j=0}^i f_{j}} \cdot 4^{i+1} \label{eq:no_speed_upper_guess_8}
  \end{align}
  which proves Lemma~\ref{lm:no_speed_upper_2}. 
\end{proof}

We are now ready to prove the main theorem about the competitive ratio of Algorithm~\ref{alg:no_speed_away}.
\begin{theorem}\label{thm:no_speed_away_upper}
  For the {\tt NoSpeed/Away} model, the competitive ratio of Algorithm~\ref{alg:no_speed_away} when applied to the sequence $f_j = 2^j$, for all $j \geq 0$, is at most
  \begin{align*}
    \begin{cases}
      5 &\text{if } v \leq \frac{1}{2} \\
      1+\frac{16 \left( \log  \frac{1}{1-v} \right)^2}{(1-v)^4} &\text{otherwise}
    \end{cases}
  \end{align*}
  where $\log$ is the base-2 logarithm.
\end{theorem}

\begin{proof}
Consider the first index $i$ such that $v_i \geq v$. 
It follows that and $v_{i-1} < v$, and so
\begin{align*}
  1-2^{-2^{i-1}} < v 
  \Rightarrow 2^{2^{i-1}} < \frac{1}{1-v} .
\end{align*}
Then, by Lemma~\ref{lm:no_speed_upper_2}, the competitive ratio of is at most
\begin{align*}
  1+ 2^{1+\sum_{j=0}^i f_{j}} \cdot 4^{i+1} 
    &= 1+ 2^{1+\sum_{j=0}^i 2^j} \cdot 4^{i+1} \\
    &= 1+2^{2^{i+1} } \cdot 4^{i+1} 
      \leq 1+\left(\frac{1}{1-v}\right)^4 \cdot \left(4 \cdot \log \left( \frac{1}{1-v} \right) \right)^2 \\ 
    &\leq 1+ \frac{16  \left( \log \frac{1}{1-v} \right)^2}{(1-v)^4}
\end{align*}
This proves Theorem~\ref{thm:no_speed_away_upper}. 
\end{proof}

\begin{remark}
  By Theorem~\ref{thm:full_knowledge_away}, $ 1 + \frac{2}{1-v}$ is a lower bound on any algorithm when both $v, d$ are known. As a consequence it must also be a lower bound when $d$ is known but $v$ is not.
\end{remark}

\subsection{The {\tt NoSpeed/Toward} Model}
\label{sec:no_speed_toward}

Now we consider the case where the target is moving toward the origin.

\begin{algorithm}[H]
  \caption{Online Algorithm for {\tt NoSpeed/Toward} Model}\label{alg:full_knowledge_toward1}
  \begin{algorithmic}[1]
    \State {{\bf input:} target initial distance $d$}
    \State{choose any direction and go for time $d$}
    \If{target not found}
      \State{change direction and go until target is found}
    \EndIf
  \end{algorithmic}
\end{algorithm}

\begin{theorem}\label{thm:no_speed_toward_upper}
  For the {\tt NoSpeed/Toward} model, Algorithm~\ref{alg:full_knowledge_toward1} achieves an optimal competitive ratio at most $3$.
\end{theorem}
\begin{proof}
  The robot chooses a direction (without loss of generality, say to the right) and goes for a time $d$ (this is where the robot makes use of its knowledge of the distance $d$). 
  If it does not find the target it changes direction. 
  In the meantime the target has moved for a distance $dv$ and now must be at location $-d + dv$. 
  Therefore at the time the robot changes direction the distance between robot and target is equal to $d - (-d + dv) = 2d -dv$, and hence the meeting will take place in additional time $\frac {2d-dv}{1+v}$. 
  It follows that the total time required for the robot to meet the target must be equal to  $d + \frac {2d-dv}{1+v}$. 
  The resulting competitive ratio satisfies
  \begin{align*}
    CR \leq \frac{d + \frac {2d-dv}{1+v}}{\frac d{v+1}} = 3.
  \end{align*}
  This proves the upper bound. 

  To prove the lower bound we argue as follows. 
  If the searcher never visits either of the points $\pm d$ then the competitive ratio is arbitrarily large for very small values of $v$. 
  Let $\epsilon > 0$ be sufficiently small and let the speed of the target be $v = \epsilon/3$. 
  Consider the first time, say $t$, that the robot reaches one of the points $\pm (d-\epsilon)$. 
  Without loss of generality let this point be $d-\epsilon$ and suppose the target is adversarially placed at $-d$.
  Then at time $t$ it will be located at $-d + tv$. 
  Therefore the distance between the robot and the target at time $t$ will be $d-\epsilon - (-d +tv) = 2d - tv -\epsilon$. 
  The time it takes for robot to find the target, then, is at least $d - \epsilon +\frac{2d - tv -\epsilon}{1+v}$ and the competitive ratio is at least
  \begin{align*}
    \frac{d - \epsilon +\frac{2d - tv -\epsilon}{1+v}}{\frac d{1+v}} 
      \geq 3 - \frac {2 \epsilon + (t+\epsilon)v}{d}    
  \end{align*}

  It follows easily that if $t \geq 3d - \epsilon$ then $CR \geq \frac{t}{ d/(1+v)} \geq 3 - 3 \epsilon$. 
  However, if $t \leq 3d - \epsilon$ then $ \frac {2 \epsilon + (t+\epsilon)v}{d} \leq 2 \epsilon + 3v \leq 3 \epsilon$, since by assumption $v = \epsilon /3$. 
  Therefore again $CR \geq 3 - 3 \epsilon$.
  This completes the proof of Theorem~\ref{thm:no_speed_toward_upper}. 
  
\end{proof}

\section{The {\tt NoKnowledge} Model}\label{sec:no_knowledge}
  
In this section we assume that neither the initial distance $d$ nor the speed $v$ of the target is known to the robot.

\subsection{The {\tt NoKnowledge/Away} Model}
\label{sec:no_knowledge_away}

We now describe an approximation strategy resembling that described in Section~\ref{sec:no_speed_away}.
For this strategy though, the robot will need to guess both the target's speed {\em and} its initial distance.

Consider the situation where neither the distance $d$ to the target nor its  speed $v<1$ is known to the robot. Also consider two monotone increasing non-negative integer sequences $\{ f_i, g_i : i \geq 0 \}$ such that $f_0 = 1$ and $g_0 = 0$ and $f_i < f_{i+1}$ and $g_i < g_{i+1}$, for all $i\geq 0$.
The idea of the algorithm is to search for the target by making a guess for its speed and starting distance in rounds as follows. 
The robot, starting from the origin, alternates searching to the right and left. 
On the $i$-th round, it guesses that the target's speed does not exceed $v_i = 1 - 2^{-f_i}$ and that it's initial distance from the origin does not exceed $2^{g_i}$.
Using these guesses, the robot searches exactly the distance required (which we will later denote $d_i$) to catch the target, given its guesses are correct {\em and} that the target is in the direction the robot searches in round $i$.
If robot does not find the target after searching this distance, it returns to the origin and begins the next round.
Later in the analysis we will show how to select the integer sequences $\{ f_i, g_i : i \geq 0\}$ so as to obtain bounds on the competitive ratio.
We formalize the algorithm described above as Algorithm~\ref{alg:no_knowledge_away}.
\begin{algorithm}[H]
  \caption{Online Algorithm for {\tt NoKnowledge/Away} Model}\label{alg:no_knowledge_away}
  \begin{algorithmic}[1]
    \State {{\bf Inputs;} 
    Integer sequences $\{ f_i, g_i : i \geq 0\}$ such that $f_i < f_{i+1}$ and $g_i < g_{i+1}$, for $i \geq 0$ and $f_0 = 1$ and $g_0 = 0$;}
    
    \State $t \gets 0$
    \For {$i \leftarrow 0, 1, 2, \ldots$ until target found}
      \State $d_i \gets 2^{g_i}$
      \State $v_i \gets 1-2^{-f_i}$
      \State $x_i \gets (-1)^i \cdot \frac{d_i + tv_i}{1-v_i}$
      \State move to $x_i$ and back to the origin
      \State $t \gets t + |x_i|$
    \EndFor
  \end{algorithmic}
\end{algorithm}

Since there is always an integer $i\ \geq 1$ such that both $v_i = 1 - 2^{-f_i} \geq v$ and $2^{g_i} \geq d$, it is clear that the robot will eventually succeed in catching the target. 
Next we analyze the competitive ratio of the algorithm. 

\begin{mylmmrep}\label{lm:no_knowledge_upper_2fg}
  For the {\tt NoKnowledge/Away} model, if Algorithm~\ref{alg:no_knowledge_away} terminates successfully in round $i+1$ then its competitive ratio must satisfy
  \begin{equation}
    CR \leq 1+ \frac{2(i+2)}{d} \cdot 2^{g_{i+1} } \cdot 2^{ \sum_{j=0}^i f_j } \cdot 4^{i+1} . \label{eq:lmquessfg}
  \end{equation}
\end{mylmmrep}
\begin{proof}
  We call each iteration of the loop in Algorithm~\ref{alg:no_knowledge_away} a {\em round}.
  For any round $i$, let $d_i$ be the distance from the origin to where the target would be if its speed was equal to $v_i = 1 - 2^{-f_i}$ and its starting position $2^{g_i}$.
  Recall that during the first $i-1$ unsuccessful rounds, the taret is moving further and further away from the origin.
  If the robot is at the origin and the speed of the target is $v_i$ then it takes time at most $x_i = \frac{d_i}{1-v_i} = 2^{f_i} d_i$ for the robot to catch up to the target, for $i \geq 0$. 
  Observe from the algorithm that $d_0 = 1$ and $v_0 = 1/2$ and
  \begin{align}
    d_i = 2^{g_i}+ 2v_i \sum_{j=0}^{i-1}  x_j . \label{eq:no_knowledge__upper_guess_dfg}
  \end{align}
  Therefore, it follows from the definition of $x_i$ that
  \begin{align}
    x_i &= \label{eq:no_knowledge__upper_guess_xfg}
    2^{f_i} \left( 2^{g_i} + 2v_i \sum_{j=9}^{i-1}  x_j \right) .
  \end{align}
  As a consequence
  \begin{align}
    \sum_{j=0}^{i-1}  x_j 
    &= \label{eq:no_knowledge__upper_guess_3} \frac{x_i - 2^{f_i + g_i} }{2^{f_i} \cdot 2 \cdot v_i}
  \end{align}
  Similarly, if we replace $i$ with $i+1$ we have that
  \begin{align}
    \sum_{j=0}^{i}  x_j 
    &= \label{eq:no_knowledge_upper_guess_4fg} \frac{x_{i+1} - 2^{f_{i+1} + g_{i+1}}}{2^{f_{i+1}} \cdot 2 \cdot v_{i+1}} .
  \end{align}
  Subtracting Equation~\eqref{eq:no_knowledge__upper_guess_3} from  Equation~\eqref{eq:no_knowledge_upper_guess_4fg} we derive the recurrence
  \begin{align}
    x_i
    &= \label{eq:no_knowledge_upper_guess_5fg} 
    \frac{x_{i+1} - 2^{f_{i+1} + g_{i+1}} }{2^{f_{i+1}} \cdot 2 \cdot v_{i+1}} -
    \frac{x_i - 2^{f_i + g_i} }{2^{f_i} \cdot 2 \cdot v_i}
  \end{align}
  Collecting similar terms and simplifying Equation~\eqref{eq:no_knowledge_upper_guess_5fg} yields
  \begin{align}
    \frac{x_{i+1}}{2^{f_{i+1}} \cdot 2 \cdot v_{i+1}}
      &= \notag \left( 1 + \frac 1{2^{f_i} \cdot 2 \cdot v_i}\right) x_i
      + \left( \frac{2^{f_{i+1} + g_{i+1}}}{2^{f_{i+1}} \cdot 2 \cdot v_{i+1}} - 
      \frac{2^{f_i+g_i}}{2^{f_i} \cdot 2 \cdot v_{i}} \right) \\
    &= \notag \left( 1 + \frac 1{2^{f_i} \cdot 2 \cdot v_i}\right) x_i
      + 2^{g_{i+1}-1} \left( 
        \frac{1}{v_{i+1}} - 
        \frac{2^{g_i}}{2^{g_{i+1}} \cdot v_{i}} 
      \right) \\
    &\leq \label{eq:no_knowledge_upper_guess_5_newfg} 
      \left( 1 + \frac 1{2^{f_i} \cdot 2 \cdot v_i}\right) x_i + 2^{g_{i+1} -1 }
  \end{align}
  where Inequality~\eqref{eq:no_knowledge_upper_guess_5_newfg} follows since $\frac{1}{v_{i+1}} - \frac{2^{g_i}}{2^{g_{i+1}} \cdot v_{i}} \leq 1$.

  If we multiply out with the denominator in the lefthand side of Inequality~\eqref{eq:no_knowledge_upper_guess_5_newfg} and simplify the righthand side we derive  the following recursive inequalities
  \begin{align}
    x_{i+1} &\leq \notag
    2^{f_{i+1}} \cdot 2 \cdot v_{i+1} \left( 1+ \frac 1{2^{f_i} \cdot 2 \cdot v_i}\right) x_i + 2^{f_{i+1} + g_{i+1}} v_{i+1}\\
    &\leq \label{eq:no_knowledge_upper_guess_6referfg}
    \left( 2 (2^{f_{i+1}}-1) +2^{f_{i+1}-f_i} \cdot \frac{v_{i+1}}{v_i}\right) x_i + 2^{f_{i+1} + g_{i+1}} \\
    &\leq \label{eq:no_knowledge_upper_guess_6fg} 
    2^{f_{i+1} } \cdot 4 \cdot x_i + 2^{f_{i+1} + g_{i+1}} ,
  \end{align}
  where in the derivation of Inequality~\eqref{eq:no_knowledge_upper_guess_6fg} from the previous Inequality~\eqref{eq:no_knowledge_upper_guess_6referfg} we used the fact that $\frac{v_{i+1}}{v_i} \leq 2$.

  By repeated application of the last Recurrence~\eqref{eq:no_knowledge_upper_guess_6fg} above and using the fact that by definition $x_0 = 2^{f_0}d$, it follows easily by induction that
  \begin{align}
    x_{i+1} 
    &\leq \notag 2^{f_{i+1} } \cdot 4 \cdot x_i + 2^{f_{i+1} + g_{i+1}} \\
    &\leq \notag 2^{f_{i+1}+f_i} \cdot 4^2 \cdot x_{i-1}  + 2^{f_{i+1}+f_i+g_i} \cdot 4^1 + 2^{f_{i+1} + g_{i+1}} \\
    &~ \vdots \notag \\
    &\leq \notag 2^{g_0+ \sum_{j=0}^{i+1} f_j} \cdot 4^{i+1} \cdot x_0 +  2^{g_1+ \sum_{j=1}^{i+1} f_j} \cdot 4^{i} + 2^{g_2+ \sum_{j=2}^{i+1} f_j} \cdot 4^{i-1} \\
    & \notag ~~~+ \cdots +  2^{f_{i+1} + g_{i+1}} \\
    &= \label{eq:no_knowledge_guess}
    \sum_{k = 0}^{i+1}  2^{g_k+ \sum_{j=k}^{i+1} f_j} \cdot 4^{i-k+1} 
  \end{align}
  since $x_0=1$.

  The total time that has elapsed from the start until the beginning of last round $i$ (when the robot visits the origin for the last time before catching the target) will be $\sum_{j=0}^i 2x_j$ at which time the target is at distance $d + v \sum_{j=0}^i 2x_j$ from the origin. 
  As a consequence the competitive ratio of Algorithm~\ref{alg:no_knowledge_away} must satisfy the inequality
  \begin{align}
    CR &\leq \label{eq:no_knowledge_upper_cr}
    \frac{2 \sum_{j=0}^i x_j + \frac{d+ 2 v \sum_{j=0}^i x_j}{1-v}}{\frac d{1-v}} .
  \end{align}
  Simplifying the righthand side of Inequality~\eqref{eq:no_knowledge_upper_cr} and using Identity~\eqref{eq:no_knowledge__upper_guess_3} yields
  \begin{align}
    CR &\leq \notag
    1+ \frac2d \sum_{j=0}^i x_j \\
    &\leq \notag
    1 + \frac{x_{i+1}}d \cdot \frac{1}{2^{f_{i+1}} \cdot v_{i+1}}
    \mbox{~(Use Equation~\eqref{eq:no_knowledge__upper_guess_3})}\\ 
    & \leq \notag
    1+ \frac{1}{v_{i+1} d 2^{f_{i+1}} }  \sum_{k = 0}^{i+1}  2^{g_k+ \sum_{j=k}^{i+1} f_j} \cdot 4^{i-k+1}
    \mbox{~(Use Equation~\eqref{eq:no_knowledge_guess})} \\
    & \leq \notag 
    1+ \frac{1}{v_{i+1} d }  \sum_{k = 0}^{i+1}  2^{g_k+ \sum_{j=k}^{i} f_j} \cdot 4^{i-k+1} .
  \end{align}
  Since $v_{i+1} \geq 1/2$ we conclude with
  \begin{align}
    CR &\leq \notag
    1 +  \frac{2}{d }  \sum_{k = 0}^{i+1}  2^{g_k+ \sum_{j=k}^{i} f_j} \cdot 4^{i-k+1}\\
    & \leq \notag
    1 + \frac{2}{d }  \sum_{k = 0}^{i+1}  2^{g_k+ \sum_{j=k}^{i} f_j} \cdot 4^{i-k+1} \\
    & \leq \label{eq:no_knowledge_upper_guess_8fg}
    1 + \frac{2(i+2)}{d } \cdot 2^{g_{i+1} } \cdot 2^{ \sum_{j=0}^i f_j } \cdot 4^{i+1}
  \end{align}
  This completes the proof of Lemma~\ref{lm:no_knowledge_upper_2fg}.
  
\end{proof}

We now prove the following theorem.
\begin{theorem}\label{thm:no_knowledge_away_upper_alt}
  For the {\tt NoKnowledge/Away} model, Algorithm~\ref{alg:no_knowledge_away} with the sequences $g_i = f_i = 2^i$ has a competitive ratio of at most
  \begin{align*}
    1 + \frac{16}{d} \left[ \log \log \max\left(d, \frac{1}{1-v}\right) + 3\right] \cdot \max\left(d, \frac{1}{1-v}\right)^{8}
      \cdot \log^2 \max\left(d, \frac{1}{1-v}\right)
  \end{align*}
  where $\log$ is the base-2 logarithm.
\end{theorem}
\begin{proof}
  Observe that if the robot finds the target in round $i+1$, then by design, one or both of the robot's round $i-1$ guesses for the target's speed ($1-2^{-2^{i-1}}$) or initial distance ($2^{2^{i-1}}$) must have been too low, otherwise the robot would have found the target in an earlier round.
  In other words, either $1-2^{-2^{i-1}} < v$ or $2^{2^{i-1}} < d$.
  It follows, then that $i-1 < \log \log \max \left( d, \frac{1}{1-v} \right)$.
  Then by Lemma~\ref{lm:no_knowledge_upper_2fg}, an upper bound on the competitive ratio is given by 
  \begin{align*}
    CR &\leq 1 + \frac{2(i+2)}{d} \cdot 2^{g_{i+1}} \cdot 2^{ \sum_{j=0}^i f_j } \cdot 4^{i+1} \\
    &= 1 + \frac{2(i+2)}{d} \cdot 2^{2^{i+1}} \cdot 2^{ 2^{i+1}-1 } \cdot 4^{i+1} \\ 
    &= 1 + \frac{(i-1)+3}{d} \cdot \left(2^{2^{i-1}}\right)^8 \cdot 16 \left(2^{i-1}\right)^2 \\ 
    &= 1 + \frac{16}{d} \left[ \log \log \max\left(d, \frac{1}{1-v}\right) + 3\right] \cdot \max\left(d, \frac{1}{1-v}\right)^{8}
    \cdot \log^2 \max\left(d, \frac{1}{1-v}\right)
  \end{align*}
  which proves Theorem~\ref{thm:no_knowledge_away_upper_alt}.
\end{proof}

\begin{remark}
  Observe that a lower bound of $1+8\frac{1+v}{(1-v)^2}$ follows directly from the \linebreak {\tt NoDistance/Away} model.
\end{remark}

\subsection{The {\tt NoKnowledge/Toward} Model}
\label{sec:no_knowledge_toward}

We can prove the following theorem. 

\begin{mythmrep}\label{thm:no_knowledge_toward}
The optimal competitive ratio is $1 + \frac 1v$ and is given by the waiting Algorithm.
\end{mythmrep}
\begin{proof}
  The upper bound follows directly from Theorem~\ref{thm:waiting}.
  For the lower bound, consider an algorithm where the robot does not wait forever and instead moves a distance $d^\prime > 0$ to the right (without loss of generality -- a symmetric argument for the case where the robot moves to the left follows trivially) after waiting at the origin for time $t \geq 0$.
  Then consider the scenario where the target with speed $v = \frac{d}{t+d^\prime}$ is initially at $-d$ for any distance $d \geq 1$.
  Thus, the target reaches the origin at exactly the time the robot reaches $d^\prime$ and so their earliest possible meeting time is
  \begin{align*}
    t+d^\prime + \frac{d^\prime}{1+v} = \frac{d}{v} + \frac{d^\prime}{1+v} \geq \frac{d}{v}
  \end{align*}
  Thus, the competitive ratio is at least
  \begin{align*}
    \frac{d/v}{d/(1+v)} = 1+\frac{1}{v}
  \end{align*}
  This proves Theorem~\ref{thm:no_knowledge_toward}.
\end{proof}

\section{Conclusion}
We considered linear search for an autonomous robot searching for an oblivious moving target on an infinite line. Two scenarios were analyzed depending on whether the target is moving towards or away from the origin (and this is known to the robot). In either of these two scenarios we considered the knowledge the robot has about the speed and starting distance of the target. For each scenario we gave search algorithms and analyzed their competitive ratio for the four possible cases arising. Our bounds are tight in all cases when the target is moving towards the origin. They are also shown to be tight when the target is moving away from the origin and its speed $0 < v < 1$ is known to the robot; for this scenario we also obtain upper bounds when $v$ is not known to the robot. It remains an open problem to prove tight bounds for the case when $v$ is unknown to the robot and the target is moving away from the origin. It also remains open to find tight bounds for the case where the direction of movement of the target is unknown.

\bibliography{refs}

\end{document}